\documentclass[12pt,a4paper]{article}

\usepackage{amsfonts,amsmath,amssymb,amsthm}
\usepackage{latexsym,amscd}
\usepackage{amsbsy}
\usepackage{CJK}
\usepackage{indentfirst,mathrsfs}
\usepackage{amsfonts,amsmath,amssymb,amsthm}
\usepackage{latexsym,amscd}
\usepackage{amsbsy}
\usepackage{color}
\usepackage{multirow}
\usepackage{makecell}
\usepackage{float}
\usepackage{cases}
\usepackage[caption=false]{subfig}

\numberwithin{equation}{section}

\newtheorem{thm}{Theorem}
\newtheorem{lem}{Lemma}

\newtheorem{prop}{Proposition}
\newtheorem{definition}{Definition}

\newcommand{\Tr}{{\rm {Tr}}}

\begin{document}

\title{A class of APcN power functions over finite fields of even characteristic}
\author{Ziran Tu, Xiangyong Zeng, Yupeng Jiang and Xiaohu Tang
\thanks{The authors are with the Hubei Key Laboratory of Applied Mathematics, Faculty of Mathematics and Statistics, Hubei
University, Wuhan, 430062, China. Yupeng Jiang is with the School of Cyber Science and Technology, Beihang University, Beijing, 100191, China. Email: tuziran@aliyun.com, xzeng@hubu.edu.cn, jiangyupeng@amss.ac.cn, xhutang@swjtu.edu.cn}
}
\date{}
\maketitle

\begin{quote}
{\small {\bf Abstract:}
In this paper, we investigate the power functions $F(x)=x^d$ over the finite field $\mathbb{F}_{2^{4n}}$, where $n$ is a positive integer and $d=2^{3n}+2^{2n}+2^{n}-1$. It is proved that $F(x)=x^d$ is APcN at certain $c$'s in $\mathbb{F}_{2^{4n}}$, and it is the second class of APcN power functions over finite fields of even characteristic. Further,  the $c$-differential spectrum of these power functions is also determined.
}

\end{quote}

\section{Introduction} \label{intro}

Let $n$, $m$ be two positive integers and $\mathbb{F}_{2^n}$ denote the finite field with $2^n$ elements. An S-box is a vectorial Boolean function from  $\mathbb{F}_{2^n}$ to $\mathbb{F}_{2^m}$, also called an $(n,m)$-function. The security of most modern block ciphers deeply relies on cryptographic properties of their S-boxes since S-boxes usually are the only nonlinear elements of these cryptosystems. It is therefore significant to employ S-boxes with good cryptographic properties in order to resist various kinds of cryptanalytic attacks.

Differential attack \cite{BS} is one of the most fundamental cryptanalytic approaches targeting symmetric-key primitives and is the first statistical attack for breaking iterated block ciphers.
The differential uniformity of S-boxes, which was introduced by Nyberg in \cite{N}, can be used to measure how well the S-box used in a cipher could resist the differential attack.
\begin{definition}\label{defn1}
Let $\mathbb{F}_{q}$ be the finite field of $q$ elements. A function $F$ defined over $\mathbb{F}_{q}$ is called differentially $\delta_{F}$-uniform, where $\delta_{F}=\mathop{max}\limits_{a\in \mathbb{F}_{q}^{*},b\in \mathbb{F}_{q}}\delta_{F}(a,b)$ and
$$\delta_{F}(a,b)=\#\{x\in \mathbb{F}_{q}:F(x+a)-F(x)=b\}.$$
\end{definition}
We call the function $F$ perfect nonlinear  (PN) or almost perfect nonlinear (APN), if $\delta_{F}=1$ or $\delta_{F}=2$, respectively.
It is well-known that PN functions only exists for an odd prime power $q$. Thus, when $q$ is even, APN functions have the best resistance to differential attacks.
To analyze the ciphers using modular multiplication as primitive operations more effectively, the authors in \cite{borisov} proposed the concept of multiplicative differential.
Very recently, based on this new type of differential, Ellingsen, Felke, Riera, St\v{a}nic\v{a} and Tkachenko gave the definition of $c$-differential uniformity in \cite{Ellingsen-it}:
\begin{definition}\label{defn2}
Let $q$ be a prime power and $\mathbb{F}_{q}$ be the finite field with $q$ elements. Given a function $F: \mathbb{F}_{q}\rightarrow \mathbb{F}_{q}$, the (multiplicative) $c$-derivative of $F$ with respect to $a$ is defined as
$$_{c}D_{a}F(x)=F(x+a)-cF(x).$$
Denote
$$_c\Delta_{F}(a,b)=\#\{x\in \mathbb{F}_{q}:{_cD_{a}F(x)=b}\}$$
and
$$_c\Delta_{F}=\mathop{max}\limits_{a,b\in \mathbb{F}_{q}} {_c\Delta_{F}(a,b)}.$$
Then $F$ is called differentially $(c, \,_{c}\Delta_{F})$-uniform.
\end{definition}
Note that if $c=0$ or $a=0$, then $_cD_{a}F(x)$ is just a shift of $F(x)$ or trivially $(1-c)F(x)$. If $c=1$, then $_cD_{a}F(x)$ becomes the usual derivative and the $c$-differential uniformity becomes differential uniformity in Definition \ref{defn1}. Similarly, we call a function perfect $c$-nonlinear (PcN) or almost perfect $c$-nonlinear (APcN), if $_c\Delta_{F}=1$ or $_c\Delta_{F}=2$, respectively. It is worth noting that PcN functions exist for even $q$, which is a big difference between PN and PcN properties. So far as we know, there are only very few results about PcN and APcN functions. The $c$-differential property of some power functions including Inverse functions, Gold functions, etc., have been investigated \cite{hasan,zhengdabin,yanhaode,zhazhengbang}. In \cite{zhengdabin} the authors give a necessary and sufficient condition for the Gold functions to be PcN, they further conjectured that all the PcN functions in $\mathbb{F}_{2^n}$ are linear functions, Gold functions and their inverses. Several ideas including the AGW criterion, cyclotomic method, the perturbing and swapping method \cite{bartoli,stanica,wuyanan} have been used to construct functions with low $c$-differential uniformity.

In this paper, we prove that this special power permutation $F(x)=x^{d}$ over $\mathbb{F}_{2^{4n}}$ is APcN on $c\in \mathbb{F}_{2^{4n}}\setminus\{0,1\}$ satisfying $c^{1+2^{2n}}=1$, where $d=2^{3n}+2^{2n}+2^n-1$. By introducing two parameters $\alpha$ and $\beta$ satisfying $\alpha^{q+1}=\beta^{q+1}=1$, we transform the APcN problem into solving a two-equation system on $\alpha$ and $\beta$. Then  a new variable $u$ is used to induce an equation with algebraic degree four, which help us give the final proof. To the best of our knowledge, there are only two classes of APcN power functions over the finite fields with even characteristic, the first one is the well-known Inverse functions, the second one is the power functions proposed in this paper.

\section{Preliminaries}\label{pres}
Let $q$ be a prime power, $\mathbb{F}_{q^m}$ and $\mathbb{F}_{q^n}$ are two finite fields with $m\mid n$. Then $\mathbb{F}_{q^m}$ can be seen as a subfield of $\mathbb{F}_{q^n}$ and the {\it relative trace} from $\mathbb{F}_{q^n}$ to $\mathbb{F}_{q^m}$ is defined as
$$\Tr_{m}^{n}(x)=x+x^{q^m}+x^{q^{2m}}+\cdots+x^{q^{\left(\frac{n}{m}-1\right)m}}.$$
If $m=1$, we call the above trace {\it Absolute}.
Given a finite field $\mathbb{F}_{q}$, let $s$ be a positive integer and $s\mid (q-1)$, define
\begin{eqnarray*}
  \mu_s=\{x\in\mathbb{F}_{q}^{*}: x^s =1\},
\end{eqnarray*}
which is constituted by all $s$th root of unity in $\mathbb{F}_{q}$. A very important such set is the {\it unit circle} of $\mathbb{F}_{q^2}$ when $q=2^m$, which is exactly defined as
\begin{equation*}
\mu_{q+1}=\left\{x\in \mathbb{F}_{q^2}^{*}: x^{q+1}=1\right\}.
\end{equation*}
The following lemma describes exactly the conditions that a quadratic equation has one or two solutions in the {\it unit circle}.
\begin{lem}\label{lem0}{\rm\cite{ztu}}
Let $n=2m$ be an even positive integer and $a,b\in \mathbb{F}_{2^n}^{*}$ satisfy ${\rm Tr}_{1}^{n}\left(\frac{b}{a^2}\right)=0$. Then the
quadratic equation $x^2+ax+b=0$ has

{\rm (1)} both solutions in the unit circle if and only if
$$\begin{array}{c}
b=a^{1-2^m} \text{ and }
{\rm Tr}_{1}^{m}\left(\frac{b}{a^2}\right)={\rm Tr}_{1}^{m}\left(\frac{1}{a^{1+2^m}}\right)=1;
\end{array}
$$

{\rm (2)} exactly one solution in the unit circle, if and only if
$$
\begin{array}{c}
b\ne a^{1-2^m} \text{ and }
(1+b^{1+2^m})(1+a^{1+2^m}+b^{1+2^m})+a^2b^{2^m}+a^{2^{m+1}}b=0.
\end{array}
$$
\end{lem}
An important fact about the unit circle is the polar-decomposition of elements, i.e.,
each element in $\mathbb{F}_{q^2}^{*}$ can be uniquely written as
$$x=\lambda y,$$
where $\lambda\in \mu_{q+1}$ and $y\in \mathbb{F}_{q}^{*}$.
Let $F(x)=x^d$ be a power function on $\mathbb{F}_{q}$. Note that  ${_c\Delta_{F}(a, b)}={_c\Delta_{F}(1, b/a^d)}$, where ${_c\Delta_{F}(a, b)}$ is defined as Definition \ref{defn2}. Hence the differential characteristics of $F(x)=x^d$ are completely determined by the values of ${_c\Delta_{F}(1, b)}$ for $b\in\mathbb{F}_{q}$.  Let $\omega_{c,i}$ be defined as follows:
\begin{eqnarray*}
\omega_{c,i}=\#\{b\in\mathbb{F}_{q}: {_c\Delta_{F}(1, b)}=i\}.
\end{eqnarray*}
The $c$-differential spectrum of $F(x)=x^d$ at the point $c$ is the set $\mathbb{S}$ of $\omega_{c,i}$ with $0\leq i\leq {_c\Delta_F}$:
\begin{eqnarray*}
\mathbb{S}=\{\omega_{c,0},\omega_{c,1},\cdots, \omega_{c,{_c\Delta_F}}\}.
\end{eqnarray*}

\section{The APcN property of the power functions}
From now on, we always assume that $n$ is a positive integer, $q=2^n$ and $d=q^3+q^2+q-1$.
We further assume that $c\in \mu_{q^2+1}\setminus\{1\}$, and we will prove that at all such $c$'s the power function $x^d$ is APcN. Note that for such $c$, we have $\Tr_{n}^{4n}(c)=c+c^q+c^{-1}+c^{-q}\ne 0$, due to that $c$ and $c^{-1}$ are two solutions in $\mathbb{F}_{q^4}\setminus\mathbb{F}_{q^2}$ of the quadratic equation $x^2+(c+c^{-1})x+1=0$. The polar decomposition $b=\lambda z$ for any $b\in \mathbb{F}_{q^4}^{*}$, with $\lambda\in \mu_{q^2+1}$ and $z\in \mathbb{F}_{q^2}$, will be frequently used in later discussions, and we always denote $v=b^q+b^{q^2}$ for convenience.
To be preparations, several useful propositions are given as follows:
\begin{prop}\label{prop1}

{\rm (1)} For $v\in \mathbb{F}_{q^4}$, if ${\Tr}_{n}^{4n}(cv^{1+q})\cdot{\rm Tr}_{n}^{4n}(c^{1+q}v^{2q})\ne 0$, then
$$\Tr_{1}^{n}\left(\frac{\Tr_{n}^{4n}(cv^{1+q})}{\Tr_{n}^{4n}(c^{1+q}v^{2q})}\right)=1.$$

{\rm (2)} For $b\in \mathbb{F}_{q^4}$ satisfying $b^{1+q^2}=1$, $\Tr_{n}^{4n}(cv^{1+q})=0$ indicates
$b\in \{1,c\}$.
\end{prop}
\begin{proof}
(1) Let $c=d^2$ for some $d\in \mathbb{F}_{q^4}$, then $c^{q^2+1}=1$ if and only if $d^{1+q^2}=1$. The denominator
\begin{eqnarray*}
   &&\Tr_{n}^{4n}(c^{1+q}v^{2q})=\left(\Tr_{n}^{4n}(d^{1+q}v^{q})\right)^2  \\
   &=& (d^{q+1}v^q+d^{q^2+q}v^{q^2}+d^{q^3+q^2}v^{q^3}+d^{1+q^3}v)^2 \\
   &=& (\nu+\nu^q)^2,
\end{eqnarray*}
where $\nu\triangleq d^{1+q^3}v+d^{q^2+q}v^{q^2}\in \mathbb{F}_{q^2}\setminus \mathbb{F}_{q}$ due to $\Tr_{n}^{4n}(c^{1+q}v^{2q})\ne 0$. The numerator $\Tr_{n}^{4n}(cv^{q+1})$ equals
\begin{eqnarray*}
   &&  \Tr_{n}^{4n}(d^2v^{q+1})\\
   &=& d^2v^{q+1}+d^{2q}v^{q^2+q}+d^{2q^2}v^{q^3+q^2}+d^{2q^3}v^{1+q^3} \\
   &=& (d^{1+q^3}v+d^{q^2+q}v^{q^2})(d^{q+1}v^q+d^{q^3+q^2}v^{q^3})\\
   &=& \nu^{1+q}.
\end{eqnarray*}
Then
$$
   \Tr_{1}^{n}\left(\frac{\nu^{1+q}}{(\nu+\nu^q)^2}\right)=1.$$

   (2) From (1), $\Tr_{n}^{4n}(cv^{1+q})=0$ if and only if $$\nu=d^{1+q^3}(b^q+b^{q^2})+d^{q^2+q}(b^{q^3}+b)=0,$$
  in which the middle part equals
   \begin{eqnarray*}
      &&d^{1-q}(b^q+b^{-1})+d^{-1+q}(b^{-q}+b)  \\
      &=& d^{1-q}\left(b^q+b^{-1}+d^{2(q-1)}(b^{-q}+b)\right) \\
      &=& d^{1-q}\left(b^q+b^{-1}+c^{q-1}(b^{-q}+b)\right)\\
      &=& d^{1-q}b^{-q-1}(1+b^{q+1})(b^q+c^{q-1}b),
   \end{eqnarray*}
   then $b\in\{1,c\}$ due to the facts $b^{1+q^2}=c^{1+q^2}=1$ and ${\rm gcd}(q^2+1,q^2-1)=1$.
\end{proof}
\begin{prop}\label{prop2}
For all $b\in \mathbb{F}_{q^4}$, define
\begin{equation}\label{eq4-1}
\begin{array}{cll}
C_{1}(k)&=&(1+k^4)(1+b^{q+q^3})+(k^3+k)(c^q+c^{-q}+(c+c^{-1})b^{q+q^3}) \\
C_{0}(k)&=&\Tr_{n}^{4n}(b)+k\left(\Tr_{n}^{4n}(bc)+\Tr_{n}^{4n}(b)\Tr_{n}^{4n}(c)\right)\\
     &&+k^2\left(\Tr_{n}^{4n}(b)+\Tr_{n}^{4n}(c^{1+q}(b^{q^2}+b^{q^3}))\right)+k^3\Tr_{n}^{4n}(cb^{q^2}).
\end{array}
\end{equation}
If $C_{1}(k)=C_{0}(k)=0$ for some $k\in \mu_{q+1}$, then $k=1$.
\end{prop}
\begin{proof}
The proof is proceeded as follows:

  If $1+b^{1+q^2}=0$, then $C_{1}(k)=k(1+k^2)(c^q+c^{-q}+c+c^{-1})=0$ gives that $k=1$, due to $c^q+c^{-q}+c+c^{-1}=\Tr_{n}^{4n}(c)\ne0$ for $c\in \mu_{q^2+1}\setminus\{1\}$. The assumptions $1+b^{1+q^2}\ne 0$ and $c^q+c^{-q}+(c+c^{-1})b^{q+q^3}=0$ also indicate that $k=1$. Now assume that $1+b^{1+q^2}\ne 0$ and $c^q+c^{-q}+(c+c^{-1})b^{q+q^3}\ne0$. If there exists some $k\in \mu_{q+1}\setminus\{1\}$ such that $C_{1}(k)=0$, then $(k^2+1)(1+b^{q+q^3})+k(c^q+c^{-q}+(c+c^{-1})b^{q+q^3})=0$, which means that
  $$u\triangleq k+k^{-1}=\frac{c^q+c^{-q}+(c+c^{-1})b^{q+q^3}}{1+b^{q+q^3}}$$
  satisfies $u\in \mathbb{F}_{q}^{*}$. Note that $u=u^q$ if and only if
  $$\frac{c^q+c^{-q}+(c+c^{-1})b^{q+q^3}}{1+b^{q+q^3}}=\frac{c+c^{-1}+(c^q+c^{-q})b^{1+q^2}}{1+b^{1+q^2}},$$
  by expanding the above equality we have
  \begin{eqnarray*}
     && (1+b^{1+q^2})(c^q+c^{-q}+(c+c^{-1})b^{q+q^3})+(1+b^{q+q^3})(c+c^{-1}+(c^q+c^{-q})b^{1+q^2})  \\
     &=& (c+c^{-1}+c^{q}+c^{-q})(1+b^{1+q+q^2+q^3})=0,
  \end{eqnarray*}
  which implies $b^{1+q+q^2+q^3}=1$. By substituting $b$ with $\lambda z$, we have
  \begin{eqnarray*}
     u^{1+q}&=&\frac{(c^q+c^{-q}+(c+c^{-1})b^{q+q^3})(c+c^{-1}+(c^q+c^{-q})b^{1+q^2})}{(1+b^{q+q^3})(1+b^{1+q^2})}  \\
     &=&\frac{(c+c^{-1})^2b^{q+q^3}+(c^q+c^{-q})^2b^{1+q^2}}{b^{q+q^3}+b^{1+q^2}} \\
     &=&\frac{(c+c^{-1})^2z^{2q}+(c^q+c^{-q})^2z^2}{z^{2q}+z^2},
  \end{eqnarray*}
  which gives
  \begin{equation}\label{eq4-2}
  u=\frac{(c+c^{-1})z^q+(c^q+c^{-q})z}{z^q+z}.
  \end{equation}
  We substitute $k^2=uk+1$ into $C_{0}(k)$ and obtain that
  \begin{eqnarray*}
    C_{0}(k)&=&\Tr_{n}^{4n}(b)+k\left(\Tr_{n}^{4n}(bc)+\Tr_{n}^{4n}(b)\Tr_{n}^{4n}(c)\right)\\
     &&+(uk+1)\left(\Tr_{n}^{4n}(b)+\Tr_{n}^{4n}(c^{1+q}(b^{q^2}+b^{q^3}))\right)+\left((u^2+1)k+u\right)\Tr_{n}^{4n}(cb^{q^2})  \\
     &=&\Tr_{n}^{4n}\left(c^{1+q}(b^{q^2}+b^{q^3})\right)+u\Tr_{n}^{4n}(cb^{q^2})\\
     &&+k\left(\Tr_{n}^{4n}(bc)+\Tr_{n}^{4n}(b)\Tr_{n}^{4n}(c)+u\left(\Tr_{n}^{4n}(b)+\Tr_{n}^{4n}(c^{1+q}(b^{q^2}+b^{q^3}))\right)\right.\\
       &&~~~~~+\left.(1+u^2)\Tr_{n}^{4n}(cb^{q^2})\right).
  \end{eqnarray*}
Since all the coefficients belong to $\mathbb{F}_{q}$ and $k\notin \mathbb{F}_{q}$, we have that $C_{0}(k)=0$ if and only if
\begin{equation}\label{eq4-3}
\Tr_{n}^{4n}\left(c^{1+q}(b^{q^2}+b^{q^3})\right)+u\Tr_{n}^{4n}(cb^{q^2})=0
\end{equation}
and
\begin{eqnarray}
&&\Tr_{n}^{4n}(bc)+\Tr_{n}^{4n}(b)\Tr_{n}^{4n}(c)+u\Tr_{n}^{4n}(b)+\Tr_{n}^{4n}(cb^{q^2})\nonumber\\
&=&\Tr_{n}^{4n}\left(c(b^q+b^{q^3})\right)+u\Tr_{n}^{4n}(b)=0.\label{eq4-4}
\end{eqnarray}
By plugging \eqref{eq4-2} into \eqref{eq4-4}, we get
\begin{eqnarray*}
   && \Tr_{n}^{4n}\left(c(b^q+b^{q^3})\right)(z^q+z)+((c+c^{-1})z^q+(c^q+c^{-q})z)\Tr_{n}^{4n}(b)  \\
   &=&\Tr_{n}^{4n}\left(c(\lambda^q+\lambda^{-q})z^q\right)(z^q+z)+((c+c^{-1})z^q+(c^q+c^{-q})z)\Tr_{n}^{4n}(\lambda z)  \\
   &=&\left((\lambda^q+\lambda^{-q})z^q(c+c^{-1})+(\lambda+\lambda^{-1})z(c^q+c^{-q})\right)(z+z^q)\\
   &&+\left((c+c^{-1})z^q+(c^q+c^{-q})z\right)\left((\lambda+\lambda^{-1})z+(\lambda^q+\lambda^{-q})z^q\right)\\
   &=&z^{q+1}\left((\lambda^q+\lambda^{-q})(c+c^{-1})+(\lambda+\lambda^{-1})(c^q+c^{-q})\right.\\
   &&+\left.(\lambda+\lambda^{-1})(c+c^{-1})+(\lambda^q+\lambda^{-q})(c^q+c^{-q})\right)\\
   &=&z^{q+1}\Tr_{n}^{4n}(c)\Tr_{n}^{4n}(\lambda)=0,
\end{eqnarray*}
which means that $\Tr_{n}^{4n}(\lambda)=0$ and then $\lambda=1$. That is to say, we have $b=z\in \mathbb{F}_{q^2}^{*}$. Then \eqref{eq4-3} gives \begin{eqnarray*}
   && \Tr_{n}^{4n}\left(c^{1+q}(z+z^{q})\right)(z+z^q)+\left((c+c^{-1})z^q+(c^q+c^{-q})z\right)\Tr_{n}^{4n}(cz) \\
   &=&\Tr_{n}^{4n}(c^{1+q})(z+z^q)^2+\left((c+c^{-1})z^q+(c^q+c^{-q})z\right)\left((c+c^{-1})z+(c^q+c^{-q})z^q\right)  \\
   &=&z^{q+1}\Tr_{n}^{4n}(c^2)\neq 0,
\end{eqnarray*}
which contradicts.

  This completes the proof.
\end{proof}

\begin{prop}\label{prop3}
Assume that $\Tr_{n}^{4n}\left(c(v+v^{q})\right)+\Tr_{n}^{4n}\left(c^{q+1}v^{q}\right)=0$. Define
\begin{eqnarray*}
  A_{1} &=& \Tr_{n}^{4n}\left(c(v+v^{q})\right), \\
  A_{2} &=& c^{q}+c^{-q}+(c+c^{-1})b^{q+q^3}.
\end{eqnarray*}
Then for any $b\in \mathbb{F}_{q^4}$, $A_{1}$ and $A_{2}$ are not zero at the same time.
\end{prop}
\begin{proof}
Firstly observe that if $b=0$, then $v=0$, $A_{1}=0$ and $A_{2}=c^{q}+c^{-q}\ne 0$. Similarly $b^{q^2+1}=1$ gives that $A_{2}=c^{q}+c^{-q}+c+c^{-1}\ne0$. If $b=b^{q^2}$ and $b\notin \{0,1\}$, it suffices to show $A_{2}\ne0$ due to $A_{1}=0$. Since $A_{2}=0$ means $b^{1+q+q^2+q^3}=1$, which gives that $b^{q+1}=1$, and then
\begin{eqnarray*}
   A_{1}&=& \Tr_{n}^{4n}(c^{q+1}v^q)=\Tr_{n}^{4n}\left(c^{q+1}(b^{q^2}+b^{q^3})\right)=\Tr_{n}^{4n}\left(c^{q+1}(b+b^{-1})\right) \\
   &=& (b+b^{-1})\Tr_{n}^{4n}\left(c^{q+1}\right)\ne 0,
\end{eqnarray*}
due to $b+b^{-1}\in \mathbb{F}_{q}^{*}$. Now assume that $b\notin  \mathbb{F}_{q^2}\cup \mu_{q^2+1}$, the decomposition $b=\lambda z$ indicates that $\lambda\ne 1$ and $z\ne 1$. Note that $A_{2}=0$ gives
\begin{equation}\label{eq4-5}
z^{2q}=\frac{c^{q}+c^{-q}}{c+c^{-1}}.
\end{equation}
If $A_{1}=0$, then we have
\begin{eqnarray*}
   && (c+c^{-1})(b^{q}+b^{q^3})+(c^{q}+c^{-q})(b+b^{q^2}) \\
   &=& (c+c^{-1})(\lambda^{q}+\lambda^{-q})z^q+(c^q+c^{-q})(\lambda+\lambda^{-1})z=0.
\end{eqnarray*}
By squaring the second equality and plugging into which with \eqref{eq4-5}, we obtain that
$$ (c+c^{-1})^2(\lambda^{q}+\lambda^{-q})^2\frac{c^{q}+c^{-q}}{c+c^{-1}}+(c^q+c^{-q})^2(\lambda+\lambda^{-1})^2\frac{c+c^{-1}}{c^q+c^{-q}} =0,$$
and then $\lambda+\lambda^{-1}=\lambda^{q}+\lambda^{-q}$, which contradicts with the assumption $\lambda\in \mu_{q^2+1}\setminus\{1\}$.
\end{proof}
\begin{prop}\label{prop4}
For $v\in \mathbb{F}_{q^4}$, if $\Tr_{n}^{4n}\left(c(v+v^q)\right)=\Tr_{n}^{4n}\left(c^{1+q}v^{q}\right)=0$, then
$$\Tr_{n}^{4n}(cv^{1+q})=\Tr_{n}^{4n}(c^{1+q}v^{2q})=0.$$
\end{prop}
\begin{proof}
By Proposition \ref{prop1}, it suffices to show $\Tr_{n}^{4n}(cv^{1+q})=0$. The proposition obviously holds if $v=0$. Assume $0\ne v=\rho y$ with $\rho\in \mu_{q^2+1}$ and $y\in \mathbb{F}_{q^2}^{*}$.  The conditions $\Tr_{n}^{4n}(v)=\Tr_{n}^{4n}\left(c(v+v^q)\right)=\Tr_{n}^{4n}\left(c^{1+q}v^{q}\right)=0$ give that
\begin{equation}\label{eq4-8}
(\rho+\rho^{-1})y+(\rho^{q}+\rho^{-q})y^q=0,
\end{equation}
\begin{equation}\label{eq4-9}
(c\rho+c^{-1}\rho^{-1}+c^q\rho^{-1}+c^{-q}\rho)y+(c^{q}\rho^q+c^{-q}\rho^{-q}+c\rho^q+c^{-1}\rho^{-q})y^q=0
\end{equation}
and
\begin{equation}\label{eq4-10}
(c^{q-1}\rho^{-1}+c^{-q+1}\rho)y+(c^{q+1}\rho^{q}+c^{-1-q}\rho^{-q})y^q=0.
\end{equation}
Note that if $\rho=1$, then \eqref{eq4-9} gives $y=y^{q}$, and then \eqref{eq4-10} does not hold due to $c^{q-1}+c^{-q+1}+c^{q+1}+c^{-q-1}=(c+c^{-1})(c^q+c^{-q})\ne0$. So we must have $\rho\ne 1$. By equations \eqref{eq4-8} and \eqref{eq4-9}, we get
\begin{eqnarray*}
   && (\rho+\rho^{-1})(c^{q}\rho^q+c^{-q}\rho^{-q}+c\rho^q+c^{-1}\rho^{-q})
   +(\rho^q+\rho^{-q})(c\rho+c^{-1}\rho^{-1}+c^q\rho^{-1}+c^{-q}\rho)  \\
   &=& \rho^{q+1}(c^q+c^{-q})+\rho^{q-1}(c+c^{-1})+\rho^{-q+1}(c+c^{-1})+\rho^{-q-1}(c^{q}+c^{-q}) \\
   &=& \Tr_{n}^{4n}\left((c+c^{-1})\rho^{q-1}\right)=0.
\end{eqnarray*}
The combination of \eqref{eq4-8} and \eqref{eq4-10} gives that
$$(\rho+\rho^{-1})(c^{1+q}\rho^q+c^{-1-q}\rho^{-q})+(\rho^q+\rho^{-q})(c^{q-1}\rho^{-1}+c^{-q+1}\rho)=0,$$
which is equivalent to
$$\Tr_{n}^{4n}((\rho+\rho^{-1})c^{1+q}\rho^q)=\Tr_{n}^{4n}((c+c^{-1})c^{q}\rho^{q-1})=0.$$
Denote $\xi=\rho^{q-1}\in \mu_{q^2+1}\setminus\{1\}$, from $\Tr_{n}^{4n}\left((c+c^{-1})\xi\right)=0$ and $\Tr_{n}^{4n}((c+c^{-1})c^{q}\xi)=0$, we have
\begin{eqnarray*}
  (c+c^{-1})(\xi+\xi^{-1}) &=& \alpha  \\
  (c+c^{-1})(c^q\xi+c^{-q}\xi^{-1}) &=&\beta
\end{eqnarray*}
for some $\alpha,\beta\in \mathbb{F}_{q}$. Obviously $\alpha\ne0$. Then
$${c^q\xi+c^{-q}\xi^{-1}}=\frac{\beta}{\alpha}({\xi+\xi^{-1}}),$$
together with $c^{q}+c^{-q}=\frac{\alpha}{\xi^{q}+\xi^{-q}}$, we get
$c^q(\xi+\xi^{-1})=\frac{\beta}{\alpha}(\xi+\xi^{-1})+\frac{\alpha\xi^{-1}}{\xi^{q}+\xi^{-q}}$, i.e,
$$c=\frac{\alpha\xi^{q}}{(\xi+\xi^{-1})(\xi^q+\xi^{-q})}+\frac{\beta}{\alpha}.$$
If $\beta\ne0$, the condition $c^{q^2+1}=1$ implies
\begin{eqnarray*}
   1&=&\left(\frac{\alpha\xi^{q}}{(\xi+\xi^{-1})(\xi^q+\xi^{-q})}+\frac{\beta}{\alpha}\right)
       \left(\frac{\alpha\xi^{-q}}{(\xi+\xi^{-1})(\xi^q+\xi^{-q})}+\frac{\beta}{\alpha}\right)  \\
   &=&  \frac{\alpha^2}{(\xi+\xi^{-1})^2(\xi^q+\xi^{-q})^2}+\frac{\beta^2}{\alpha^2}+\frac{\beta}{\xi+\xi^{-1}}
\end{eqnarray*}
and then $\xi+\xi^{-1}\in \mathbb{F}_{q}$, which contradicts. This proves that $\beta=0$, i.e, $c^q\xi\in \mathbb{F}_{q^2}$, which actually indicates $c^q\xi=1$ due to $c^{q}\xi\in \mu_{q^2+1}$.
Then
\begin{eqnarray*}
   && \Tr_{n}^{4n}(cv^{1+q})=\Tr_{n}^{4n}(c\rho^{1+q}y^{1+q})=y^{1+q}\Tr_{n}^{4n}(c\rho^{1+q}) \\
   &=& y^{1+q}\Tr_{n}^{4n}(c^q\rho^{q-1})=y^{1+q}\Tr_{n}^{4n}(c^q\xi)=0.
\end{eqnarray*}
This completes the proof.
\end{proof}
The following is about the factorization of a polynomial $G(u)\in \mathbb{F}_{q}[u]$, which is an important observation for the main proof.
\begin{prop}\label{prop5}
For any $b\in \mathbb{F}_{q^4}$, the polynomial
\begin{equation}\label{eq4-11}
G(u)=G_{0}+G_{1}u+G_{2}u^2+G_{3}u^3+G_{4}u^4
\end{equation}
where
\begin{eqnarray*}
  G_{0} &=& \Tr_{n}^{4n}\left(c^2(b^{2q}+b^{2q^3})\right)+\Tr_{n}^{4n}\left(c^{2(1+q)}(b^{2q^2}+b^{2q^3})\right) \\
  G_{1} &=& \Tr_{n}^{4n}\left(c(b^q+b^{q^3})\right)\Tr_{n}^{4n}\left(c^{1+q}(b^{q^2}+b^{q^3})\right) \\
  G_{2} &=&  \Tr_{n}^{4n}(b^2)+\Tr_{n}^{4n}\left(c^2(b^{2q^2}+b^{q+q^2}+b^{q^2+q^3}+b^{q+q^3})\right)\\
           &&+\Tr_{n}^{4n}(c^{1+q})(1+b^{1+q+q^2+q^3})+\Tr_{n}^{4n}\left(c^{1+q}(b^{q+q^2}+b^{1+q^3})\right)\\
  G_{3} &=&  \Tr_{n}^{4n}(c)(1+b^{1+q+q^2+q^3})+\Tr_{n}^{4n}(b)\Tr_{n}^{4n}(cb^{q^2})\\
  G_{4} &=&  (1+b^{q+q^3})(1+b^{1+q^2})
\end{eqnarray*}
has a factorization as
$$G(u)=\left(u^2+\Tr_{n}^{4n}(c)u+\Tr_{n}^{4n}(c^{1+q})\right)\left((1+b^{q+q^3})(1+b^{1+q^2})u^2+Au+B\right),$$
where
\begin{eqnarray*}
  A &=& \Tr_{n}^{4n}\left(c(b^q+b^{q^2})(b^{q^2}+b^{q^3})\right)=\Tr_{n}^{4n}(cv^{q+1}) \\
  B &=&  \Tr_{n}^{4n}\left(c^{1+q}(b^{2q^2}+b^{2q^3})\right)=\Tr_{n}^{4n}(c^{1+q}v^{2q}).
\end{eqnarray*}
\end{prop}
\begin{proof}
The expression of $G(u)$ seems a bit complicated and the factorization can be verified directly, so we give the proof in Appendix A.
\end{proof}
In the next we give the last Proposition in this paper, the proof of which is a bit technical and is critical to complete the final proof.
\begin{prop}\label{prop6}
Assume that $b^{1+q^2}\ne 1$, $\Tr_{n}^{4n}(c^{1+q}v^{2q})=0$ and $\Tr_{n}^{4n}(cv^{q+1})\ne0$, and denote
\begin{equation}\label{eq4-12}
\Omega=\frac{\Tr_{n}^{4n}(c^{q+1})(1+b^{1+q+q^2+q^3})+\Tr_{n}^{4n}(c^2b^{q+q^3})}{\Tr_{n}^{4n}(c^{q+1}v^q)^2}
   +\frac{(1+b^{q+q^3})(1+b^{1+q^2})}{\Tr_{n}^{4n}(cv^{q+1})},
\end{equation}
then $\Omega\in \mathbb{F}_{q}$ and $\Tr_{1}^{n}(\Omega)=1$.
\end{prop}
\begin{proof}
Obviously $\Omega\in \mathbb{F}_{q}$. We note that $\Tr_{n}^{4n}(c^{1+q}v^{2q})=0$ is equivalent to
\begin{equation}\label{neweq}
\Tr_{n}^{4n}(c(v+v^{q}))+\Tr_{n}^{4n}(c^{q+1}v^q)=0,
\end{equation}
which comes from
\begin{eqnarray*}
   && \Tr_{n}^{4n}\left(c(b^q+b^{q^3})\right)+\Tr_{n}^{4n}\left(c^{1+q}(b^{q^2}+b^{q^3})\right) \\
   &=& \Tr_{n}^{4n}\left((c+c^{-1})(c^q+1)b^{q^3}\right)
\end{eqnarray*}
and \begin{eqnarray*}
       && \Tr_{n}^{4n}\left((c+c^{-1})(c^q+1)b^{q^3}\right)^2 \\
       &=&  \Tr_{n}^{4n}\left((c+c^{-1})^2(c^{2q}+1)b^{2q^3}\right) \\
       &=& \Tr_{n}^{4n}\left((c+c^{-1})^2(c^q+c^{-q})c^qb^{2q^3}\right)\\
       &=& (c+c^{-1})(c^{q}+c^{-q})\Tr_{n}^{4n}\left((c+c^{-1})c^qb^{2q^3}\right)\\
       &=& (c+c^{-1})(c^{q}+c^{-q})\Tr_{n}^{4n}\left(c^{q+1}(b^{2q^2}+b^{2q^3})\right).
    \end{eqnarray*}
There is a relation between the two denominators in the expression of $\Omega$ as
 \begin{equation}\label{eq4-13}
 \Tr_{n}^{4n}(c^{q+1}v^q)^2=\Tr_{n}^{4n}(c^{q+1})\Tr_{n}^{4n}(cv^{q+1}).
 \end{equation}
 We write $\Tr_{n}^{4n}(c(v+v^{q})=\alpha+\alpha^q$, where $\alpha\triangleq(c+c^{-1})(v+v^q)\in \mathbb{F}_{q^2}$, then
\begin{eqnarray*}
   v^q&=&\frac{\alpha}{c+c^{-1}}+v  \\
   v^{q^2}&=&\frac{\alpha^q}{c^{q}+c^{-q}}+\frac{\alpha}{c+c^{-1}}+v   \\
   v^{q^3}&=&\frac{\alpha^q}{c^{q}+c^{-q}}+v
\end{eqnarray*}
and \begin{eqnarray*}
       && \Tr_{n}^{4n}(c^{q+1}v^q)
       =c^{q+1}v^q+c^{-1+q}v^{q^2}+c^{-q-1}v^{q^3}+c^{1-q}v  \\
       &=& c^{q+1}\left(\frac{\alpha}{c+c^{-1}}+v\right)+c^{-1+q}\left(\frac{\alpha^q}{c^{q}+c^{-q}}+\frac{\alpha}{c+c^{-1}}+v\right)
         +c^{-q-1}\left(\frac{\alpha^q}{c^{q}+c^{-q}}+v\right)+c^{1-q}v \\
       &=&\Tr_{n}^{4n}(c^{q+1})v+\frac{\alpha}{c+c^{-1}}(c^{q+1}+c^{-1+q})+\frac{\alpha^q}{c^q+c^{-q}}(c^{-1+q}+c^{-q-1}) \\
       &=&\Tr_{n}^{4n}(c^{q+1})v+\alpha c^q+\alpha^q c^{-1}.
           \end{eqnarray*}
The condition \eqref{neweq} can be rewritten as $\alpha+\alpha^q=\Tr_{n}^{4n}(c^{q+1})v+\alpha c^q+\alpha^q c^{-1}$,
which gives that
\begin{equation}\label{eq4-14}
v = \frac{\alpha+\alpha^q+\alpha c^q+\alpha^{q}c^{-1}}{\Tr_{n}^{4n}(c^{q+1})},
\end{equation}
then
  $$v^q =\frac{\alpha+\alpha^q+\alpha^q c^{-1}+\alpha c^{-q}}{\Tr_{n}^{4n}(c^{q+1})},$$
and
\begin{eqnarray*}
  &&v^{q+1}\Tr_{n}^{4n}(c^{q+1})^2 \\
  &=&(\alpha+\alpha^q+\alpha c^q+\alpha^qc^{-1})(\alpha+\alpha^q+\alpha^qc^{-1}+\alpha c^{-q})  \\
   &=&(\alpha+\alpha^q)^2+(\alpha+\alpha^q)(\alpha c^q+\alpha c^{-q})+(\alpha c^q+\alpha^qc^{-1})(\alpha^qc^{-1}+\alpha c^{-q})  \\
  &=& \alpha^{2q}+(\alpha+\alpha^q)\alpha( c^q+c^{-q})+\alpha^{q+1}(c^{q-1}+c^{-q-1})+\alpha^{2q}c^{-2}\\
  &=& c^{-1}(c+c^{-1})\alpha^{2q}+(\alpha+\alpha^q)\alpha(c^q+c^{-q})+\alpha^{q+1}(c^q+c^{-q})c^{-1}.
\end{eqnarray*}
We obtain that
\begin{eqnarray*}
  &&\Tr_{n}^{4n}(cv^{q+1})\Tr_{n}^{4n}(c^{q+1})^2\\
  &=&\Tr_{n}^{4n}\left((\alpha+\alpha^q)\alpha(c^q+c^{-q})c\right)  \\
  &=&(\alpha+\alpha^q)\Tr_{n}^{4n}(\alpha(c^q+c^{-q})c)  \\
  &=&(\alpha+\alpha^q)^2\Tr_{n}^{4n}(c^{q+1}),
\end{eqnarray*}
which indicates \eqref{eq4-13}. By the way, we have $\alpha\notin \mathbb{F}_{q}$ due to the assumption $\Tr_{n}^{4n}(cv^{q+1})\ne 0$. Now \eqref{eq4-12} becomes
$$\Omega=\frac{\Tr_{n}^{4n}(c^2b^{q+q^3})+\Tr_{n}^{4n}(c^{q+1})(b^{q+q^3}+b^{1+q^2})}{(\alpha+\alpha^q)^2}.$$
From $b=\lambda z$ and $b^{q}+b^{q^3}=\frac{\alpha}{c+c^{-1}}$, we get $z^q(\lambda^q+\lambda^{-q})=\frac{\alpha}{c+c^{-1}}$, i.e,
$$z^{q}=\frac{\alpha}{(c+c^{-1})(\lambda^q+\lambda^{-q})},$$
then
\begin{equation}\label{eq4-15}
b^q=\frac{\alpha\lambda^q}{(c+c^{-1})(\lambda^q+\lambda^{-q})}
\end{equation}
From
\begin{equation*}
  b^{q+q^3}=z^{2q}=\frac{\alpha^2}{(c+c^{-1})^2(\lambda^q+\lambda^{-q})^2},
\end{equation*}
we see that
$$\Tr_{n}^{4n}(c^2b^{q+q^3})=\Tr_{n}^{4n}\left(\frac{c^2\alpha^2}{(c+c^{-1})^2(\lambda^q+\lambda^{-q})^2}\right)
=\frac{\alpha^2}{(\lambda^q+\lambda^{-q})^2}+\frac{\alpha^{2q}}{(\lambda+\lambda^{-1})^2}$$
and
\begin{eqnarray*}
  &&\Tr_{n}^{4n}(c^{q+1})(b^{q+q^3}+b^{1+q^2}) \\
  &=&\Tr_{n}^{4n}(c^{q+1})\left(\frac{\alpha^2}{(\lambda^{q}+\lambda^{-q})^2(c+c^{-1})^2}
     +\frac{\alpha^{2q}}{(\lambda+\lambda^{-1})^2(c^q+c^{-q})^2}\right)   \\
   &=&\frac{\alpha^2}{\lambda^{2q}+\lambda^{-2q}}\frac{c^q+c^{-q}}{c+c^{-1}}+\frac{\alpha^{2q}}{\lambda^2+\lambda^{-2}}\frac{c+c^{-1}}{c^q+c^{-q}}.
\end{eqnarray*}
From the two equalities we obtain the nominator of $\Omega$ equals
\begin{eqnarray}
   && \frac{\alpha^2}{\lambda^{2q}+\lambda^{-2q}}\left(1+\frac{c^q+c^{-q}}{c+c^{-1}}\right)
    +\frac{\alpha^{2q}}{\lambda^2+\lambda^{-2}}\left(1+\frac{c+c^{-1}}{c^q+c^{-q}}\right) \nonumber\\
   &=& \Tr_{n}^{4n}(c)\left(\frac{\alpha^2}{\lambda^{2q}+\lambda^{-2q}}\frac{1}{c+c^{-1}}
    +\frac{\alpha^{2q}}{\lambda^2+\lambda^{-2}}\frac{1}{c^q+c^{-q}}\right). \label{eq4-16}
\end{eqnarray}
Recall that $v=b^q+b^{q^2}$, then \eqref{eq4-14} and \eqref{eq4-15} give
$$\frac{\alpha\lambda^q}{(c+c^{-1})(\lambda^q+\lambda^{-q})}+\frac{\alpha^q\lambda^{-1}}{(c^q+c^{-q})(\lambda+\lambda^{-1})}
=\frac{\alpha+\alpha^q+\alpha c^q+\alpha^qc^{-1}}{\Tr_{n}^{4n}(c^{q+1})},$$
which is equivalent to
\begin{eqnarray*}
   && \alpha\left(\frac{1+c^q}{\Tr_{n}^{4n}(c^{q+1})}+\frac{\lambda^q}{(c+c^{-1})(\lambda^q+\lambda^{-q})}\right)
   =\alpha^q\left(\frac{1+c^{-1}}{\Tr_{n}^{4n}(c^{q+1})}+\frac{\lambda^{-1}}{(c^q+c^{-q})(\lambda+\lambda^{-1})}\right) \\
   &\Leftrightarrow& \alpha \frac{(1+c^q)(\lambda^q+\lambda^{-q})+\lambda^q(c^q+c^{-q})}{\Tr_{n}^{4n}(c^{q+1})(\lambda^q+\lambda^{-q})}
   =\alpha^q\frac{(1+c^{-1})(\lambda+\lambda^{-1})+\lambda^{-1}(c+c^{-1})}{\Tr_{n}^{4n}(c^{q+1})(\lambda+\lambda^{-1})}.
\end{eqnarray*}
We get
\begin{equation*}
  \alpha^q=\frac{D^q}{D}\left(\frac{\lambda+\lambda^{-1}}{\lambda^q+\lambda^{-q}}\right)\alpha,
\end{equation*}
where $D=\lambda+\lambda^{-1}+\lambda c^{-1}+\lambda^{-1}c$. By plugging the above equality into \eqref{eq4-16}, we simplify the expression of $\Omega$ as
\begin{eqnarray*}
  \Omega &=&\frac{\Tr_{n}^{4n}(c)\frac{\alpha^2}{\lambda^{2q}+\lambda^{-2q}}(\frac{1}{c+c^{-1}}+\frac{D^{2q}}{D^2}\frac{1}{c^q+c^{-q}})}
    {\alpha^2(1+\frac{D^{2q}}{D^2}\frac{\lambda^2+\lambda^{-2}}{\lambda^{2q}+\lambda^{-2q}})}  \\
&=&\frac{\Tr_{n}^{4n}(c)}{\Tr_{n}^{4n}(c^{q+1})}\frac{D^2(c^q+c^{-q})+D^{2q}(c+c^{-1})}{D^{2}(\lambda^q+\lambda^{-q})^2+D^{2q}(\lambda+\lambda^{-1})^2}\\
&=&\frac{\Tr_{n}^{4n}(c)\Tr_{n}^{4n}(c^{-1}\lambda^2)}{\Tr_{n}^{4n}((\lambda^q+\lambda^{-q})\lambda^{-1}c)^2},
\end{eqnarray*}
due to
\begin{eqnarray*}
   && D^2(c^q+c^{-q})+D^{2q}(c+c^{-1})  \\
   &=& (\lambda+\lambda^{-1}+\lambda^{-1}c+\lambda c^{-1})^2(c^q+c^{-q})+(\lambda^q+\lambda^{-q}+\lambda^{-q}c^q+\lambda^{q}c^{-q})^2(c+c^{-1}) \\
   &=& \Tr_{n}^{4n}((c^q+c^{-q})(c^{-2}+1)\lambda^2) \\
   &=&(c+c^{-1})(c^q+c^{-q})\Tr_{n}^{4n}(c^{-1}\lambda^2)
\end{eqnarray*}
and
\begin{eqnarray*}
   &&D(\lambda^q+\lambda^{-q})+D^{q}(\lambda+\lambda^{-1})   \\
   &=&(\lambda+\lambda^{-1}+\lambda^{-1}c+\lambda c^{-1})(\lambda^q+\lambda^{-q})+(\lambda^q+\lambda^{-q}+\lambda^{-q}c^q+\lambda^{q}c^{-q})(\lambda+\lambda^{-1})\\
   &=&(\lambda^{-1}c+\lambda c^{-1})(\lambda^q+\lambda^{-q})+(\lambda^{-q}c^q+\lambda^{q}c^{-q})(\lambda+\lambda^{-1})\\
   &=&\Tr_{n}^{4n}((\lambda^q+\lambda^{-q})\lambda^{-1}c).
\end{eqnarray*}
Let $\beta=(c^{-1}+c^{-q})\lambda^{q+1}+(c+c^{q})\lambda^{-q-1}\in \mathbb{F}_{q^2}$, which does not belong to $\mathbb{F}_{q}$ from
   $$\beta+\beta^q=\Tr_{n}^{4n}((c^{-1}+c^{-q})\lambda^{q+1})=\Tr_{n}^{4n}(c^{-1}\lambda^{q+1}+c^{-1}\lambda^{1-q})
   =\Tr_{n}^{4n}((\lambda^q+\lambda^{-q})\lambda^{-1}c).$$
The fact that
\begin{eqnarray*}
  \beta^{1+q} &=&\left((c^{-1}+c^{-q})\lambda^{q+1}+(c+c^{q})\lambda^{-q-1}\right)\left((c^{-q}+c)\lambda^{q-1}+(c^{-1}+c^{q})\lambda^{-q+1}\right)   \\
   &=& \Tr_{n}^{4n}\left((c^{-1}+c^{-q})(c^{-1}+c^q)\lambda^2\right) \\
   &=& \Tr_{n}^{4n}\left((c^{-2}+1+c^{-q-1}+c^{q-1})\lambda^2\right) \\
   &=& (c+c^{-1}+c^{q}+c^{-q})\Tr_{n}^{4n}(c^{-1}\lambda^2)
\end{eqnarray*}
gives $$\Tr_{1}^{n}(\Omega)=\Tr_{1}^{n}\left(\frac{\beta^{q+1}}{(\beta+\beta^q)^2}\right)=1.$$
\end{proof}
With above propositions, now we can prove the APcN property of the power function as the following Theorem.
\begin{thm}\label{thm2}
Let $n\ge 1$ be a positive integer and $q=2^n$. Then the power permutation $x^{q^3+q^2+q-1}$ over $\mathbb{F}_{q^4}$ is an APcN permutation for each $c\in \mu_{q^2+1}\setminus\{1\}$.
\end{thm}
\begin{proof}
It is easy to verify that ${\rm gcd}(q^3+q^2+q-1,q^4-1)=1$. Since $(q^3+q^2+q-1)q^2\equiv {(q-1)(q^2+1)+2} \,{\rm mod}\, (q^4-1)$, it suffices to prove that for any $b\in \mathbb{F}_{q^4}$, the following $c$-differential equation
\begin{equation}\label{eq4-18}
(x+1)^{(q-1)(q^2+1)+2}+c^2x^{(q-1)(q^2+1)+2}=b^2
\end{equation}
has at most two solutions in $\mathbb{F}_{q^4}$. Note that $x=0$($x=1$) is a solution of \eqref{eq4-18} if and only if $b=1$($b=c$), we only need consider solutions  in $\mathbb{F}_{q^4}\setminus\{0,1\}$. Let $(x+1)^{(q-1)(q^2+1)}=\alpha^2$ and $x^{(q-1)(q^2+1)}=\beta^2$, then
\eqref{eq4-18} gives that $\alpha^2(x^2+1)+c^2\beta^2x^2=b^2$, i.e
$$x=\frac{\alpha+b}{\alpha+c\beta}.$$
We obtain the following equation system
\begin{eqnarray}
  \left(\frac{b+c\beta}{\alpha+c\beta}\right)^{(q-1)(q^2+1)} &=&\alpha^2  \label{eq4-19}\\
  \left(\frac{\alpha+b}{\alpha+c\beta}\right)^{(q-1)(q^2+1)} &=&\beta^2  \label{eq4-20}
\end{eqnarray}
with $(\alpha,\beta)\in \mu_{q+1}\times\mu_{q+1}$. It is easy to see that there exists an one-to-one correspondence between the solutions of $x$ satisfying \eqref{eq4-18} and the solution pairs $(\alpha,\beta)$ of the above system. Since the left hand side of \eqref{eq4-19} equals
$$
    \left(\frac{b^{q^2}+c^{q^2}\beta^{q^2}}{\alpha^{q^2}+c^{q^2}\beta^{q^2}}\cdot\frac{b+c\beta}{\alpha+c\beta}\right)^{q-1}
   = \frac{\frac{b^{q^3}+c^{q^3}\beta^{q^3}}{\alpha^{q^3}+c^{q^3}\beta^{q^3}}\cdot\frac{b^q+c^q\beta^q}{\alpha^q+c^q\beta^q}}
   {\frac{b^{q^2}+c^{q^2}\beta^{q^2}}{\alpha^{q^2}+c^{q^2}\beta^{q^2}}\cdot\frac{b+c\beta}{\alpha+c\beta}}
   = \frac{\frac{b^{q^3}+c^{-q}\beta^{-1}}{\alpha^{-1}+c^{-q}\beta^{-1}}\cdot\frac{b^q+c^q\beta^{-1}}{\alpha^{-1}+c^q\beta^{-1}}}
   {\frac{b^{q^2}+c^{-1}\beta}{\alpha+c^{-1}\beta}\cdot\frac{b+c\beta}{\alpha+c\beta}},
$$
we rewrite \eqref{eq4-19} as
$$\frac{b^{q+q^3}+\beta^{-2}+(c^{-q}b^q+c^qb^{q^3})\beta^{-1}}{\alpha^{-2}+\beta^{-2}+\alpha^{-1}\beta^{-1}(c^q+c^{-q})}
  =\alpha^2\cdot\frac{b^{1+q^2}+\beta^2+\beta(c^{-1}b+cb^{q^2})}{\alpha^2+\beta^2+\alpha\beta(c+c^{-1})}.$$
By expanding the above equality, we obtain our first relation on variables $\alpha,\beta$. Since this relation seems a bit long, for convenience we write the terms $\alpha^i\beta^j$ and the corresponding coefficients as the following:
$$
\begin{array}{lcl}
\hline
\alpha^2&  1+b^{q+q^3} &  \alpha^2\beta^2  \\
\alpha&    (c+c^{-1})(c^{-q}b^q+c^qb^{q^3})+(c^q+c^{-q})(c^{-1}b+cb^{q^2})&  \alpha\beta^2 \\
\alpha^2\beta^{-2}& 1+b^{1+q^2}& \alpha^2 \\
\alpha^2\beta^{-1}& c^{-q}b^q+c^qb^{q^3}+c^{-1}b+cb^{q^2} &  \alpha^2\beta \\
\alpha\beta^{-1}& b^{1+q^2}(c^q+c^{-q})+c+c^{-1}&  \alpha\beta \\
\alpha\beta &  c^q+c^{-q}+(c+c^{-1})b^{q+q^3}&  \alpha\beta^3  \\
\beta^2&  1+b^{q+q^3} &  \beta^4 \\
\beta& c^{-q}b^q+c^qb^{q^3}+c^{-1}b+cb^{q^2}  &  \beta^3  \\
1&          1+b^{1+q^2}&    \beta^2 \\
\hline
\end{array}
$$
in which the first column are these original terms in the expanding expressions and the third column is obtained from the first column by multiplying $\beta^2$. We denote $G_{1}(\alpha,\beta)=0$, which is constituted by the terms in the third column and their corresponding coefficients. By completely similar computations, we transform \eqref{eq4-20} as
$$\frac{\frac{b^{q^3}+\alpha^{-1}}{\alpha^{-1}+c^{-q}\beta^{-1}}\cdot\frac{b^q+\alpha^{-1}}{\alpha^{-1}+c^q\beta^{-1}}}
{\frac{b^{q^2}+\alpha}{\alpha+c^{-1}\beta}\cdot\frac{b+\alpha}{\alpha+c\beta}}=\beta^2,$$
which is equivalent to
$$\frac{b^{q+q^3}+\alpha^{-2}+\alpha^{-1}(b^q+b^{q^3})}{\alpha^{-2}+\beta^{-2}+\alpha^{-1}\beta^{-1}(c^q+c^{-q})}
  =\beta^2\frac{b^{1+q^2}+\alpha^2+\alpha(b+b^{q^2})}{\alpha^2+\beta^2+\alpha\beta(c+c^{-1})},$$
then we obtain our second relation $G_{2}(\alpha,\beta)=0$, according to the terms in the third column and the coefficients according to the following table
$$
\begin{array}{lcl}
\hline
\alpha^2&  1+b^{q+q^3} &  \alpha^4   \\
\alpha&    b+b^q+b^{q^2}+b^{q^3}&  \alpha^3 \\
\beta^2& 1+b^{q+q^3}& \alpha^2\beta^2  \\
\beta & (c+c^{-1})(b^q+b^{q^3})+(c^q+c^{-q})(b+b^{q^2}) &  \alpha^2\beta \\
\alpha\beta& c^q+c^{-q}+(c+c^{-1})b^{q+q^3}&  \alpha^3\beta \\
\alpha^{-2}\beta^2 &  1+b^{1+q^2}&  \beta^2 \\
\alpha^{-1}\beta&  b^{1+q^2}(c^q+c^{-q})+c+c^{-1} &  \alpha\beta \\
\alpha^{-1}\beta^2& b+b^q+b^{q^2}+b^{q^3}  &  \alpha\beta^2 \\
1&          1+b^{1+q^2}&    \alpha^2 \\
\hline
\end{array}
$$
From the fact
\begin{eqnarray*}
   && (c+c^{-1})(b^q+b^{q^3})+(c^q+c^{-q})(b+b^{q^2})+c^{-q}b^q+c^qb^{q^3}+c^{-1}b+cb^{q^2}  \\
   &=& c(b^q+b^{q^3}+b^{q^2})+c^{-1}(b^q+b^{q^3}+b)+c^{q}(b+b^{q^2}+b^{q^3})+c^{-q}(b+b^{q^2}+b^q) \\
   &=& (c+c^{-1}+c^q+c^{-q})(b+b^q+b^{q^2}+b^{q^3})+cb+c^qb^q+c^{-1}b^{q^2}+c^{-q}b^{q^3} \\
   &=& \Tr_{n}^{4n}(c)\Tr_{n}^{4n}(b)+\Tr_{n}^{4n}(cb),
\end{eqnarray*}
we write the expression $G_{1}(\alpha,\beta)+G_{2}(\alpha,\beta)$ as $H(\alpha,\beta)$, which is described as follows:
$$
\begin{array}{cll}
\hline
   1+b^{q+q^3}&  \alpha^4 & \alpha^4\\
 \Tr_{n}^{4n}(b) &  \alpha^3 &\alpha^3\\
  \Tr_{n}^{4n}(bc)+\Tr_{n}^{4n}(b)\Tr_{n}^{4n}(c)&  \alpha^2\beta &k\alpha^3\\
  \Tr_{n}^{4n}(b)+\Tr_{n}^{4n}(c^{q+1}(b^{q^2}+b^{q^3})) &  \alpha\beta^2 &k^2\alpha^3\\
 c^q+c^{-q}+(c+c^{-1})b^{q+q^3}  &  \alpha\beta^3 &k^3\alpha^4\\
 c^q+c^{-q}+(c+c^{-1})b^{q+q^3}&    \alpha^3\beta &k\alpha^4\\
  1+b^{q+q^3}     &    \beta^4  & k^4\alpha^4 \\
   \Tr_{n}^{4n}(cb^{q^2})    &     \beta^3  & k^3\alpha^3\\
\hline
\end{array}
$$
 Note that terms in the third column in the above table are obtained by substituting $\beta=k\alpha$ from the second column, which will be used in later discussions.
 If $b=1$, then $G_{2}(\alpha,\beta)=0$ indicates that $\alpha^3\beta+\alpha\beta=0$ and then $\alpha=1$, which contradicts with \eqref{eq4-20}. Similarly, if $b=c$, then $1=b^{1+q^2}$, $G_{1}(\alpha,\beta)=0$ gives $\alpha\beta+\alpha\beta^3=0$ and then $\beta=1$, which contradicts with \eqref{eq4-19}. That is to say, \eqref{eq4-18} has $x=0$($x=1$) as its unique solution for $b=1$($b=c$). In the next we assume that $b\notin \{1,c\}$.

By substituting $\beta=k\alpha$ with $k\in \mu_{q+1}$ and eliminating $\alpha^3$, we transform the equality $H(\alpha,\beta)=0$ as
\begin{equation}\label{eq4-21}
\alpha C_{1}(k)=C_{0}(k),
\end{equation}
where
$C_{0}(k)$ and $C_{1}(k)$ are exactly defined as in \eqref{eq4-1}.
We see from \eqref{eq4-21} that if $C_{1}(k)\ne0$, then $\alpha$ is uniquely determined by $k$, which means that, to determine the number of solutions of \eqref{eq4-18}, it suffices to determine the number of suitable $k$'s. If $C_{1}(k)=0$ and $C_{0}(k)\ne 0$, no such $\alpha\in \mu_{q+1}$ exists satisfying \eqref{eq4-18}. By Proposition 2, $k=1$ is the only possible element in $\mu_{q+1}$ such that $C_{0}(k)=C_{1}(k)=0$, which means that
\begin{eqnarray}
   && \Tr_{n}^{4n}(bc)+\Tr_{n}^{4n}(b)\Tr_{n}^{4n}(c)+\Tr_{n}^{4n}\left(c^{1+q}(b^{q^2}+b^{q^3})\right)+\Tr_{n}^{4n}(cb^{q^2})  \nonumber\\
   &=& \Tr_{n}^{4n}\left(c(b^q+b^{q^3})\right)+\Tr_{n}^{4n}\left(c^{1+q}(b^{q^2}+b^{q^3})\right) =0 \nonumber\\
   &\Leftrightarrow& \Tr_{n}^{4n}(c(v+v^q))+\Tr_{n}^{4n}(c^{1+q}v^q)=0,\label{eq4-22}
\end{eqnarray}
we should take into account the solution pairs $(\alpha,\alpha)$ satisfying \eqref{eq4-19} and \eqref{eq4-20}.
 Assume that $k\ne 1$ and $C_{1}(k)\ne0$. The condition $\alpha\in \mu_{q+1}$ gives that
$$C_{1}(k)^{1+q}=C_{0}(k)^{1+q}.$$
 Let $u=k+k^{-1}$, then $u^q=u$ and we will write the above equality as a polynomial on $u$. By the fact $k\in \mu_{q+1}$ and
\begin{eqnarray*}
   C_{1}(k)&=&k^2\left((k^2+k^{-2})(1+b^{q+q^3})+(k+k^{-1})(c^q+c^{-q}+(c+c^{-1})b^{q+q^3})\right)  \\
   &=&k^2u\left(u(1+b^{q+q^3})+c^q+c^{-q}+(c+c^{-1})b^{q+q^3}\right),
\end{eqnarray*}
we have
\begin{eqnarray*}
   C_{1}(k)^{1+q}&=&u^2\left(u(1+b^{q+q^3})+c^q+c^{-q}+(c+c^{-1})b^{q+q^3}\right)\\
         &&\cdot\left(u(1+b^{q^2+1})+c+c^{-1}+(c^q+c^{-q})b^{q^2+1}\right)  \\
   &=&u^4(1+b^{q+q^3})(1+b^{q^2+1})+u^3\left(\Tr_{n}^{4n}(c)(1+b^{1+q+q^2+q^3})\right)  \\
   &&+u^2\left(\Tr_{n}^{4n}(c^{q+1})(1+b^{1+q+q^2+q^3})+\Tr_{n}^{4n}(c^2b^{q+q^3})\right),
\end{eqnarray*}
due to
\begin{eqnarray}
   && (c+c^{-1}+(c^q+c^{-q})b^{q^2+1})(c^q+c^{-q}+(c+c^{-1})b^{q+q^3}) \nonumber\\
   &=& (c+c^{-1})(c^q+c^{-q})(1+b^{1+q+q^2+q^3})+(c+c^{-1})^2b^{q+q^3}+(c^q+c^{-q})^2b^{1+q^2}\nonumber\\
   &=& \Tr_{n}^{4n}(c^{q+1})(1+b^{1+q+q^2+q^3})+\Tr_{n}^{4n}(c^2b^{q+q^3}) \label{neweq2}
\end{eqnarray}
and
\begin{eqnarray*}
   &&(1+b^{q+q^3})(c+c^{-1}+(c^q+c^{-q})b^{q^2+1})+(1+b^{q^2+1})(c^q+c^{-q}+(c+c^{-1})b^{q+q^3})  \\
   &=&(1+b^{q+q^3})(c+c^{-1})+(1+b^{q+q^3})(c^q+c^{-q})b^{q^2+1}\\
      &&+(1+b^{q^2+1})(c^q+c^{-q})+(1+b^{1+q^2})(c+c^{-1})b^{q+q^3}  \\
   &=&c+c^{-1}+c^q+c^{-q}+(c+c^{-1})b^{q+q^3+1+q^2}+(c^q+c^{-q})b^{1+q^2+q+q^3}  \\
   &=& \Tr_{n}^{4n}(c)(1+b^{1+q+q^2+q^3}).
\end{eqnarray*}
The transformation of $C_{0}(k)^{q+1}$ into a polynomial on $u$ is a bit more complicated but direct. Denote
$$C_{0}(k)=E_{0}+E_{1}k+E_{2}k^2+E_{3}k^3,$$
where
$$
\begin{array}{cl}
  E_{0}=&\Tr_{n}^{4n}(b) \\
  E_{1}=& \Tr_{n}^{4n}(bc)+\Tr_{n}^{4n}(b)\Tr_{n}^{4n}(c)\\
  E_{2}=&\Tr_{n}^{4n}(b)+\Tr_{n}^{4n}(c^{1+q}(b^{q^2}+b^{q^3})) \\
  E_{3}=&\Tr_{n}^{4n}(cb^{q^2}).
\end{array}
$$
Note that all $E_{i}$'s are elements in $\mathbb{F}_{q}$, we have
\begin{eqnarray*}
  C_{0}(k)^{1+q} &=&(E_{0}+E_{1}k+E_{2}k^2+E_{3}k^3)(E_{0}+E_{1}k^{-1}+E_{2}k^{-2}+E_{3}k^{-3})  \\
   &=& E_{0}^2+E_{1}^2+E_{2}^2+E_{3}^2+(k+k^{-1})(E_{0}E_{1}+E_{1}E_{2}+E_{2}E_{3})\\
   &&+(k^2+k^{-2})(E_{0}E_{2}+E_{1}E_{3})
      +(k^3+k^{-3})E_{0}E_{3}\\
   &=& E_{0}^2+E_{1}^2+E_{2}^2+E_{3}^2+u(E_{0}E_{1}+E_{1}E_{2}+E_{2}E_{3}+E_{0}E_{3})\\
   &&+u^2(E_{0}E_{2}+E_{1}E_{3})+u^3E_{0}E_{3}
\end{eqnarray*}
because of $k^3+k^{-3}=u^3+u$. Then the task becomes the computations of all the coefficients.
Firstly we have
\begin{eqnarray*}
   &&E_{0}+E_{1}+E_{2}+E_{3}  \\
   &=& \Tr_{n}^{4n}(bc)+\Tr_{n}^{4n}(b)\Tr_{n}^{4n}(c)+\Tr_{n}^{4n}(c^{1+q}(b^{q^2}+b^{q^3}))+\Tr_{n}^{4n}(cb^{q^2}) \\
   &=& \Tr_{n}^{4n}(c(b^{q}+b^{q^3}))+\Tr_{n}^{4n}(c^{1+q}(b^{q^2}+b^{q^3})),
\end{eqnarray*}
\begin{eqnarray*}
   && E_{0}E_{1}+E_{1}E_{2}+E_{2}E_{3}+E_{0}E_{3}=(E_{1}+E_{3})(E_{0}+E_{2})  \\
   &=& \left(\Tr_{n}^{4n}(bc)+\Tr_{n}^{4n}(b)\Tr_{n}^{4n}(c)+\Tr_{n}^{4n}(cb^{q^2})\right)\Tr_{n}^{4n}\left(c^{1+q}(b^{q^2}+b^{q^3})\right)\\
   &=&\Tr_{n}^{4n}\left(c(b^q+b^{q^3})\right)\Tr_{n}^{4n}\left(c^{1+q}(b^{q^2}+b^{q^3})\right).
\end{eqnarray*}
The computations
\begin{eqnarray*}
   E_{1}E_{3}&=&\left(\Tr_{n}^{4n}(bc)+\Tr_{n}^{4n}(b)\Tr_{n}^{4n}(c)\right)\Tr_{n}^{4n}(cb^{q^2})   \\
   &=&\Tr_{n}^{4n}\left(cb^{q^2}(bc+b^qc^q+b^{q^2}c^{q^2}+b^{q^3}c^{q^3})\right)
     + \Tr_{n}^{4n}(b)\Tr_{n}^{4n}\left(cb^{q^2}(c+c^q+c^{q^2}+c^{q^3})\right) \\
   &=&\Tr_{n}^{4n}(c^2b^{1+q^2}+c^{1+q}b^{q+q^2}+b^{2q^2}+c^{1+q^3}b^{q^2+q^3}) \\
     &&+\Tr_{n}^{4n}(b)\Tr_{n}^{4n}\left((c^2+c^{1+q}+1+c^{1+q^3})b^{q^2}\right)\\
     &=&\Tr_{n}^{4n}(c^2b^{1+q^2}+c^{1+q}b^{q+q^2}+c^{1+q^3}b^{q^2+q^3})
     +\Tr_{n}^{4n}(b)\Tr_{n}^{4n}\left((c^2+c^{1+q}+c^{1+q^3})b^{q^2}\right)\\
     &=&\Tr_{n}^{4n}(c^2b^{1+q^2})+\Tr_{n}^{4n}\left(c^{1+q}(b^{q+q^2}+b^{1+q^3})\right)
       +\Tr_{n}^{4n}(b)\Tr_{n}^{4n}(c^2b^{q^2})\\
       &&+\Tr_{n}^{4n}(b)\Tr_{n}^{4n}(c^{1+q}(b^{q^2}+b^{q^3}))
\end{eqnarray*}
and
\begin{eqnarray*}
  E_{0}E_{2}&=&\Tr_{n}^{4n}(b)\left(\Tr_{n}^{4n}(b)+\Tr_{n}^{4n}(c^{1+q}(b^{q^2}+b^{q^3}))\right)  \\
\end{eqnarray*}
give $$E_{0}E_{2}+E_{1}E_{3}=\Tr_{n}^{4n}(b^2)+\Tr_{n}^{4n}\left(c^2(b^{2q^2}+b^{q+q^2}+b^{q^2+q^3})\right)
       +\Tr_{n}^{4n}\left(c^{1+q}(b^{q+q^2}+b^{1+q^3})\right).$$
Together with $E_{0}E_{3}=\Tr_{n}^{4n}(b)\Tr_{n}^{4n}(cb^{q^2})$, we obtain an equality $G(u)=0$ from $C_{0}(k)^{1+q}=C_{1}(k)^{1+q}$, which is exactly defined as Proposition \ref{prop5}, and can be factorized as two quadratic polynomials.

Observe that $c+c^{-1}$ and $c^q+c^{-q}$ are exactly the two solutions of the equation
$$u^2+\Tr_{n}^{4n}(c)u+\Tr_{n}^{4n}(c^{1+q})=0$$
in $\mathbb{F}_{q^2}\setminus\mathbb{F}_{q}$ due to Proposition \ref{prop1}. Then $G(u)=0$ for some $u\in \mathbb{F}_{q}$ if and only if
\begin{equation}\label{eq4-23}
(1+b^{q+q^3})(1+b^{1+q^2})u^2+\Tr_{n}^{4n}(cv^{q+1})u+\Tr_{n}^{4n}\left(c^{1+q}v^{2q}\right)=0.
\end{equation}

Case (i): $\Tr_{n}^{4n}\left(c^{1+q}v^{2q}\right)\ne0$. By Proposition \ref{prop1}, we have $\Tr_{n}^{4n}(cv^{q+1})\ne0$. If $b^{1+q^2}=1$, then $u=\frac{\Tr_{n}^{4n}(c^{1+q}v^q)}{\Tr_{n}^{4n}(cv^{q+1})}$ is the unique solution of \eqref{eq4-23}. If $b^{1+q^2}\ne 1$, we
assume that $u_{0}$ and $u_{1}$ are two non-zero solutions in $\mathbb{F}_{q}$ of \eqref{eq4-23}, then  $u_{0}+u_{1}=\frac{\Tr_{n}^{4n}\left(cv^{q+1}\right)}{(1+b^{q+q^3})(1+b^{1+q^2})}$, $u_{0}u_{1}=\frac{\Tr_{n}^{4n}\left(c^{1+q}v^{2q}\right)}{(1+b^{q+q^3})(1+b^{1+q^2})}$ and
$$\Tr_{1}^{n}(u_{0}^{-1}+u_{1}^{-1})
  =\Tr_{1}^{n}\left(\frac{\Tr_{n}^{4n}(cv^{q+1})}{\Tr_{n}^{4n}(c^{1+q}v^{2q})}\right)=1
$$
from Proposition \ref{prop1}. From Lemma \ref{lem0}, we know that $k+k^{-1}=u$, i.e $k^2+ku+1=0$ has two solutions in $\mu_{q+1}$ if and only if $\Tr_{1}^{n}(u^{-1})=1$. Then there exists at most one $u\in \mathbb{F}_{q}$ contributing two $k$'s such that $k+k^{-1}=u$ and $k^{1+q}=1$.

Case (ii): $\Tr_{n}^{4n}\left(c^{1+q}v^{2q}\right)=0$. Recall that $\Tr_{n}^{4n}(c^{1+q}v^{2q})=0$ is equivalent to the assumption \eqref{eq4-22}. Then there are two kinds of solution pairs for equations system \eqref{eq4-19} and \eqref{eq4-20}. One is $(\alpha,k\alpha)$ with $\alpha$ determined by \eqref{eq4-19}, and $k\in \mu_{q+1}\setminus\{1\}$ satisfying $k+k^{-1}=u\ne 0$, which comes from \eqref{eq4-23}. The other is $(\alpha,\alpha)$ for some $\alpha\in \mu_{q+1}$. Note that $G_{1}(\alpha,\alpha)=0$ gives
\begin{equation}\label{eq4-24}
A_{2}\alpha^2+A_{1}\alpha+A_{0}=0,
\end{equation}
where
\begin{eqnarray*}
  A_{2} &=&c^q+c^{-q}+(c+c^{-1})b^{q+q^3},  \\
  A_{1} &=& (c+c^{-1})(c^{-q}b^q+c^qb^{q^3})+(c^q+c^{-q})(c^{-1}b+cb^{q^2}),\\
  A_{0} &=&c+c^{-1}+(c^q+c^{-q})b^{1+q^2}.
\end{eqnarray*}
We see that $A_{2}=A_{0}^q\in \mathbb{F}_{q^2}$ and
\begin{eqnarray*}
  A_{1} &=&   \Tr_{n}^{4n}\left((c+c^{-1})c^{-q}b^q\right) \\
   &=& \Tr_{n}^{4n}\left(c^{q+1}(b^{q^2}+b^{q^3})\right)= \Tr_{n}^{4n}\left(c(b^q+b^{q^3})\right) \\
   &=& \Tr_{n}^{4n}\left(c(v+v^q)\right),
\end{eqnarray*}
which are defined as in Proposition \ref{prop3}. By Proposition \ref{prop2}, $A_{1}$ and $A_{2}$ are not zero at the same time. If $b^{1+q^2}=1$, then $\Tr_{n}^{4n}(cv^{q+1})\ne 0$ due to the assumption $b\notin \{1,c\}$ from Proposition \ref{prop1}. We also have $A_{1}=\Tr_{n}^{4n}(c(v+v^{q}))\ne 0$ from Proposition \ref{prop4}. Together with the fact $A_{2}=\Tr_{n}^{4n}(c)\ne 0$, we see that \eqref{eq4-24} has either $0$ or $2$ solutions and \eqref{eq4-23} does not give new $k$'s since $u=0$ is its unique solution.

Assume $b^{1+q^2}\ne 1$.
 If $$A_{1}=\Tr_{n}^{4n}\left(c(v+v^q)\right)=\Tr_{n}^{4n}\left(c^{q+1}v^q\right)=0,$$
then
$$\Tr_{n}^{4n}\left(cv^{q+1}\right)=\Tr_{n}^{4n}\left(c^{q+1}v^{2q}\right)=0,$$
which means that $(\alpha,\alpha)$ with $\alpha^2=A_{2}^{q-1}$ is the unique solution pair of \eqref{eq4-19} and \eqref{eq4-20}, and then \eqref{eq4-18} has exactly one solution, due to \eqref{eq4-23} does not give new $k$'s.  Assume $A_{2}A_{1}\ne0$. By the fact $\frac{A_{0}}{A_{2}}=\frac{\frac{A_{1}}{A_{2}}}{\frac{A_{1}^q}{A_{2}^q}}$ always holds, then \eqref{eq4-24} has two solutions in $\mu_{q+1}$ if and only if
$$\Tr_{1}^{n}\left(\frac{A_{0}A_{2}}{A_{1}^2}\right)=1.$$
Observe that
$$A_{0}A_{2}=\Tr_{n}^{4n}(c^{q+1})(1+b^{1+q+q^2+q^3})+\Tr_{n}^{4n}(c^2b^{q+q^3})$$
from \eqref{neweq2}. If $\Tr_{n}^{4n}(cv^{1+q})\ne 0$ and $\Tr_{n}^{4n}(c^{1+q}v^{2q})=0$, then \eqref{eq4-23} contributes two $k$'s in $\mu_{q+1}\setminus\{1\}$ if
$$\Tr_{1}^{n}\left(\frac{(1+b^{q+q^3})(1+b^{1+q^2})}{\Tr_{n}^{4n}(cv^{q+1})}\right)=1.$$
 Then by Proposition \ref{prop6}, we see that the two cases that \eqref{eq4-23} gives two $k$'s and \eqref{eq4-24} gives two $\alpha$'s do not happen at the same time,  which completes the proof.
\end{proof}

So far, there are only a few known APcN functions, especially for finite fields with even characteristic. We list these functions in the following Table \ref{tab1}.
\begin{table}[!htb]
\caption{The known APcN functions over finite fields}
\label{tab1}
\centering
\begin{tabular}{cccc}
\hline
$p$ & $F(x)$ & $_c\Delta_F$ & Ref.
\\

\hline
any & $x^{q-2}$ & $2$ &  \cite{Ellingsen-it} \\
odd & $x^{p^k+1}$ & $2$ &  \cite{Ellingsen-it,hasan,yanhaode}\\
$3$ & $x^{\frac{3^k+1}{2}}$ & $2$ &  \cite{yanhaode}\\
$2$ & $b\phi(x)+\Tr_{q}^{q^n}(g(x^q-x)),b\phi(x)+g(x^q-x)^{\frac{q^n-1}{q-1}}$ & $2$ & \cite{bartoli} \\
any & $x(\sum_{i=1}^{l-1}x^{i\frac{q-1}{l}}+u)$&  $\le 2$ & \cite{wuyanan}                                                \\
$3$ & $(x^{p^k}-x)^{\frac{q-1}{2}+p^{ik}}+a_{1}x+a_{2}x^{p^k}+a_{3}x^{p^k}$& $\le 2$ &\cite{wuyanan}\\
$2$ & $x^{p^k+1}+\gamma\Tr_{1}^{q}(x)$& $\le 2$ &\cite{wuyanan}\\
any & $x^{q+1}+a_{0}x^q+a_{1}x$ & $2$ & \cite{wuyanan}  \\
$2$ & $u\phi(x)+\sum_{i=1}^{t}g(x^q-x)^{\frac{q^n-1}{d_{i}}}$& $2$ & \cite{wuyanan}\\
$2$ & $x^{q^3+q^2+q-1}$& $2$ & this paper \\
\hline
\end{tabular}
\end{table}
It is important to investigate whether the APcN monomial in this paper is new. From \cite{hasan}, we know that the c-differential uniformity is only preserved by affine transform, while it is not invariant under EA-equivalence and CCZ-equivalence. From Table \ref{tab1}, we see that for finite fields with even characteristic, there are six classes of APcN functions. For the function $f=x(\sum_{i=1}^{l-1}x^{i\frac{q-1}{l}}+u)$, which is obtained by cyclotomic method in \cite{wuyanan}, the experiments show that many of such $f$'s are not permutations, and when $f$ is a permutation, it is not EA-equivalent with the APcN power functions here. For other cases, by comparing the algebraic degrees and the permutation properties, it is easy to obtain that the permutation monomial in this paper is not EA-equivalent to known APcN functions.

\begin{prop}\label{prop7}
For any $b\in \mathbb{F}_{q^4}$, the $c$-differential equation \eqref{eq4-18} has exactly one solution in $\mathbb{F}_{q^4}$, if and only if
$\Tr_{n}^{4n}(c^{q+1}v^{2q})=0$ and $\Tr_{n}^{4n}(c(v+v^q))=0$.
\end{prop}
\begin{proof}
Note that when $b=1,c$, \eqref{eq4-18} has exactly one solution. For $b\ne 1,c$, from the proof of Theorem \ref{thm2}, we firstly know that \eqref{eq4-18} has either $0$ or $2$ solutions if $\Tr_{n}^{4n}(c^{q+1}v^{2q})\ne 0$. To guarantee the uniqueness of the solution, it is necessary that $\Tr_{n}^{4n}(c^{q+1}v^{2q})= 0$. If $A_{1}=\Tr_{n}^{4n}(c(v+v^q))\ne 0$, we can see that \eqref{eq4-24} gives $0$ or $2$ $\alpha$'s, which implies that $\Tr_{n}^{4n}(c(v+v^q))=0$. It is easy to see that $b=1,c$ also satisfy $\Tr_{n}^{4n}(c^{q+1}v^{2q})=0$ and $\Tr_{n}^{4n}(c(v+v^q))=0$. The sufficiency is obvious.
\end{proof}
Now we can $c$-differential spectrum of this power function.
\begin{thm}
Let $\omega_{c,i}$ denote the number of $b$'s such that \eqref{eq4-18} has exactly $i$ solutions, for $i=0,1,2$. Then we have
$$\omega_{c,0}=\omega_{c,2}=\frac{q^4-q^2}{2},\,\,\omega_{c,1}=q^2.$$
\end{thm}
\begin{proof}
From the relations
\begin{eqnarray*}
  \omega_{c,0}+\omega_{c,1}+\omega_{c,2} &=& q^4,  \\
  \omega_{c,1}+2\omega_{c,2} &=& q^4,
\end{eqnarray*}
we have $\omega_{c,0}=\omega_{c2}$. Then it suffices to determine $\omega_{c,1}$. Since $\Tr_{n}^{4n}(c^{q+1}v^{2q})=0$ is equivalent to
$\Tr_{n}^{4n}(c(v+v^q))+\Tr_{n}^{4n}(c^{q+1}v^q)=0$, and by Proposition \ref{prop7}, we have
\begin{eqnarray*}
   &&\Tr_{n}^{4n}(c(v+v^q))=\Tr_{n}^{4n}(c(b^q+b^{q^3}))=\Tr_{n}^{4n}((c+c^{-1})b^{q^3})=0,  \\
   &&\Tr_{n}^{4n}(c^{q+1}v^q)=\Tr_{n}^{4n}(c^{q+1}(b^{q^2}+b^{q^3}))=\Tr_{n}^{4n}(c^q(c+c^{-1})b^{q^3}).
\end{eqnarray*}
Observe that $c+c^{-1}$ and $c^{q}(c+c^{-1})$ are $\mathbb{F}_{q}$ independent, we know that the number of $b$'s satisfying the above two equations is $q^2$, i.e $\omega_{c,1}=q^2$.
\end{proof}

\section{Concluding remarks}
In this paper, it was proved that the power function $F(x)=x^d$ over $\mathbb{F}_{q^4}$, with $q=2^n$ and $d=q^3+q^2+q-1$, is APcN at each point $c\in \mu_{q^2+1}\setminus\{1\}$, and the corresponding $c$-differential spectrum was also determined.


\section*{Appendix A:  Proof of Proposition \ref{prop5}}
\begin{proof}
The proof is obtained by directly comparing the coefficients on both sides. Firstly we have
\begin{eqnarray*}
  \Tr_{n}^{4n}(c^{1+q})\cdot B&=&\Tr_{n}^{4n}\left(c^{1+q}(b^{2q^2}+b^{2q^3})(c^{1+q}+c^{-1+q}+c^{-q-1}+c^{1-q})\right)  \\
   &=& \Tr_{n}^{4n}\left(c^{2(1+q)}+c^{2q}+1+c^2)(b^{2q^2}+b^{2q^3})\right)\\
   &=& \Tr_{n}^{4n}\left(c^{2(1+q)}(b^{2q^2}+b^{2q^3})\right)+\Tr_{n}^{4n}\left(c^2(b^{2q}+b^{2q^3})\right) \\
   &=& G_{0}.
\end{eqnarray*}
For the coefficient of the term $u$, the computations of
\begin{eqnarray*}
   \Tr_{n}^{4n}(c)\cdot B &=& \Tr_{n}^{4n}\left(c^{1+q}(b^{2q^2}+b^{2q^3})(c+c^q+c^{-1}+c^{-q})\right) \\
   &=& \Tr_{n}^{4n}\left((c^{q+2}+c^{2q+1}+c^q+c)(b^{2q^2}+b^{2q^3})\right)\\
   &=& \Tr_{n}^{4n}\left((c^{q+2}+c^{2q+1})(b^{2q^2}+b^{2q^3})\right)+\Tr_{n}^{4n}\left(c(b^{2q}+b^{2q^3})\right)
\end{eqnarray*}
and
\begin{eqnarray*}
  \Tr_{n}^{4n}(c^{1+q})\cdot A &=& \Tr_{n}^{4n}\left(c(b^q+b^{q^2})(b^{q^2}+b^{q^3})(c^{1+q}+c^{-1+q}+c^{-q-1}+c^{1-q})\right) \\
   &=&\Tr_{n}^{4n}\left((c^{q+2}+c^q+c^{-q}+c^{2-q})(b^{q+q^2}+b^{2q^2}+b^{q+q^3}+b^{q^2+q^3})\right)\\
   &=& \Tr_{n}^{4n}\left((c^{q+2}+c^{2-q})(b^{q+q^2}+b^{2q^2}+b^{q+q^3}+b^{q^2+q^3})\right)\\
   &&+\Tr_{n}^{4n}\left(c(b^{q^2+q^3}+b^{2q^3}+b^{q^3+1}+b^{1+q}+b^{2q}+b^{q+q^2})\right)\\
   &=&\Tr_{n}^{4n}\left(c^{q+2}(b^{q+q^2}+b^{2q^2}+b^{q+q^3}+b^{q^2+q^3})\right)\\
   &&+\Tr_{n}^{4n}\left(c^{2q+1}(b^{q^2+q^3}+b^{2q^3}+b^{q^2+1}+b^{q^3+1})\right)\\
   &&+\Tr_{n}^{4n}\left(c(b^{q^2+q^3}+b^{2q^3}+b^{q^3+1}+b^{1+q}+b^{2q}+b^{q+q^2})\right),
\end{eqnarray*}
give that
\begin{eqnarray*}
   && \Tr_{n}^{4n}(c)\cdot B+\Tr_{n}^{4n}(c^{1+q})\cdot A  \\
   &=& \Tr_{n}^{4n}\left(c^{q+2}(b^{q+q^2}+b^{2q^3}+b^{q+q^3}+b^{q^2+q^3})\right)
   +\Tr_{n}^{4n}\left(c^{2q+1}(b^{q^2+q^3}+b^{2q^2}+b^{q^2+1}+b^{q^3+1})\right)\\
   &&+\Tr_{n}^{4n}\left(c(b^{q^2+q^3}+b^{q^3+1}+b^{1+q}+b^{q+q^2})\right) \\
   &=& \Tr_{n}^{4n}\left(c^{q+2}(b^q+b^{q^3})(b^{q^2}+b^{q^3})\right)
   +\Tr_{n}^{4n}\left(c^{2q+1}(b+b^{q^2})(b^{q^2}+b^{q^3})\right)
   +\Tr_{n}^{4n}\left(c(b+b^{q^2})(b^q+b^{q^3})\right) \\
   &=&\Tr_{n}^{4n}\left(c^{q+1}(b^{q^2}+b^{q^3})(c(b^q+b^{q^3})+c^q(b+b^{q^2}))\right)
      +\Tr_{n}^{4n}\left(c(b+b^{q^2})(b^q+b^{q^3})\right) \\
   &=&\Tr_{n}^{4n}\left(c^{q+1}(b^{q^2}+b^{q^3})\left(\Tr_{n}^{4n}(c(b^q+b^{q^3}))+c^{-1}(b^q+b^{q^3})+c^{-q}(b^{q^2}+b)\right)\right)\\
      &&+\Tr_{n}^{4n}\left(c(b+b^{q^2})(b^q+b^{q^3})\right) \\
   &=&\Tr_{n}^{4n}\left(c^{q+1}(b^{q^2}+b^{q^3})\right)\Tr_{n}^{4n}\left(c(b^q+b^{q^3})\right)+\Tr_{n}^{4n}\left((b^{q^2}+b^{q^3})\left(c^{q}(b^q+b^{q^3})+c(b^{q^2}+b)\right)\right)\\
      &&+\Tr_{n}^{4n}\left(c(b+b^{q^2})(b^q+b^{q^3})\right)  \\
    &=&\Tr_{n}^{4n}\left(c^{q+1}(b^{q^2}+b^{q^3})\right)\Tr_{n}^{4n}\left(c(b^q+b^{q^3})\right) \\
    &=&G_{1}.
\end{eqnarray*}
The coefficient of the term $u^2$
\begin{eqnarray*}
   && \Tr_{n}^{4n}(c)\cdot A +B+\Tr_{n}^{4n}(c^{1+q})(1+b^{q+q^3})(1+b^{1+q^2})  \\
   &=& \Tr_{n}^{4n}\left((c+c^q+c^{-1}+c^{-q})c(b^q+b^{q^2})(b^{q^2}+b^{q^3})\right)+\Tr_{n}^{4n}\left(c^{1+q}(b^{2q^2}+b^{2q^3})\right)\\
   &&+\Tr_{n}^{4n}\left(c^{1+q}(1+b^{q+q^3}+b^{1+q^2}+b^{1+q+q^2+q^3})\right)\\
   &=&\Tr_{n}^{4n}\left((c^2+c^{q+1}+1+c^{-q+1})(b^q+b^{q^2})(b^{q^2}+b^{q^3})\right)+\Tr_{n}^{4n}\left(c^{1+q}(b^{2q^2}+b^{2q^3})\right)\\
   &&+\Tr_{n}^{4n}\left(c^{1+q}(1+b^{q+q^3}+b^{1+q^2}+b^{1+q+q^2+q^3})\right)\\
      &=&\Tr_{n}^{4n}\left((c^2+1)(b^q+b^{q^2})(b^{q^2}+b^{q^3})\right)+\Tr_{n}^{4n}\left(c^{1+q}(b^{q^2}+b^{q^3})(b+b^q+b^{q^2}+b^{q^3})\right)\\
   &&+\Tr_{n}^{4n}\left(c^{1+q}(b^{2q^2}+b^{2q^3}+1+b^{q+q^3}+b^{1+q^2}+b^{1+q+q^2+q^3})\right)\\
   &=&\Tr_{n}^{4n}\left(c^2(b^q+b^{q^2})(b^{q^2}+b^{q^3})\right)+\Tr_{n}^{4n}(b^2)+
   \Tr_{n}^{4n}\left(c^{1+q}(1+b^{1+q^3}+b^{q+q^2}+b^{1+q+q^2+q^3})\right)\\
   &=&G_{2}.
\end{eqnarray*}
The coefficient of $u^3$ is
\begin{eqnarray*}
   &&\Tr_{n}^{4n}(c)(1+b^{q+q^3})(1+b^{1+q^2})+A  \\
   &=&\Tr_{n}^{4n}(c(1+b^{q+q^3}+b^{1+q^2}+b^{1+q+q^2+q^3}))+\Tr_{n}^{4n}(c(b^{q+q^2}+b^{2q^2}+b^{q+q^3}+b^{q^2+q^3})) \\
   &=&\Tr_{n}^{4n}(c(1+b^{1+q^2}+b^{1+q+q^2+q^3}+b^{q+q^2}+b^{2q^2}+b^{q^2+q^3})) \\
   &=&G_{3}
\end{eqnarray*}
and the coefficient of $u^4$ is obvious.
\end{proof}

\end{document}